\documentclass[doublecol]{epl2}

\usepackage{amsthm}
\usepackage{amsmath}
\usepackage{latexsym}
\usepackage{amsfonts}
\usepackage{amssymb}
\usepackage{color}
\usepackage{bbm,dsfont}
\usepackage{graphicx}
\usepackage{subfigure}

\usepackage{mathrsfs}
\usepackage{mathbbol}

\newtheorem{theorem}{Theorem}

\newtheorem{proposition}{Proposition}

\newcommand{\real}{\mathbb R} 
\newcommand{\complex}{\mathbb C} 
\newcommand{\integer}{\mathbb Z} 
\newcommand{\half}{\tfrac{1}{2}} 
\newcommand{\mo}[1]{\left| #1 \right|} 

\newcommand{\hi}{\mathcal{H}} 
\newcommand{\ip}[2]{\left\langle\,#1\,|\,#2\,\right\rangle} 
\newcommand{\kb}[2]{|#1\rangle\langle#2|} 
\newcommand{\tr}[1]{\textrm{tr}\left[#1\right]} 
\newcommand{\id}{I} 

\newcommand{\Mo}{\mathsf{M}}
\newcommand{\No}{\mathsf{N}}
\newcommand{\To}{\mathsf{T}}

\newcommand{\M}{\mathfrak{M}}

\newcommand{\Tmf}{\mathfrak{T}}

\newcommand{\sfm}{\mathsf{M}}

\newcommand{\sfg}{\mathsf{G}}

\title{Comparing the degrees of incompatibility inherent in probabilistic physical theories}
\shorttitle{Degrees of incompatibility in probabilistic physical theories}

\author{Paul Busch\inst{1} \and Teiko Heinosaari\inst{2} \and Jussi Schultz\inst{3} \and Neil Stevens\inst{1}}

\shortauthor{P.~Busch \etal}

\institute{
\inst{1} Department of Mathematics, University of York, York, YO10 5DD, United Kingdom\\
\inst{2} Turku Centre for Quantum Physics, Department of Physics and Astronomy, University of Turku, FI-20014 Turku, Finland\\
\inst{3} Dipartimento di Matematica, Politecnico di Milano, Piazza Leonardo da Vinci 32, I-20133 Milano, Italy
}


\abstract{
We introduce a new way of quantifying  the degrees of incompatibility of two observables in a probabilistic physical theory and,
based on this, a global measure of the {degree of incompatibility inherent in such theories, across all observable pairs.} 
This {opens up a novel} and flexible way of comparing 
probabilistic theories with respect to the nonclassical feature of incompatibility, raising many interesting questions, some of which
will be answered here. We show that quantum theory contains observables that are as incompatible as any probabilistic physical 
theory can have if arbitrary pairs of observables are considered. If one adopts a more refined measure of {the degree of incompatibility}, for instance, 
by restricting the comparison to binary observables, it turns out that there are probabilistic theories {whose inherent degree 
of incompatibility is greater than that of} quantum mechanics.  
\\
\phantom{}\\
To be published in: {\em Europhysics Letters} (2013).
}

\begin{document}
\maketitle

\thispagestyle{empty}


Quantum theory has a number of important features not known in classical physics, ranging from the 
superposition and indeterminacy principles formulated by the pioneers  to the more recently discovered 
no-cloning and no-broadcasting theorems. It is an old problem to identify operationally significant properties 
of quantum theory that distinguish it from other probabilistic theories. 
In recent years many features have been under intensive investigation from this perspective, 
including information processing \cite{Barrett07}, optimal state discrimination \cite{KiMiIm09},  
entropy \cite{Barnumetal10}, purification \cite{ChDaPe11} and discord \cite{Perinotti12}. It has been 
found that some properties are quite generally valid in any non-classical (no-signaling) probabilistic 
theories while others are specifically quantum. 

The existence of pairs of incompatible observables marks one of the most striking distinctions between 
quantum theory and  classical physical theories. There are many manifestations of incompatibility, 
perhaps the most famous being the Heisenberg uncertainty principle \cite{BuHeLa07}. However, there
are many nonclassical probabilistic theories which also possess incompatible observables, and it will be of
interest to compare quantum theory with alternative theories with respect to the feature of incompatibility.

To this end, we define the joint measurability region of any given pair of observables in a probabilistic theory. 
The joint measurability region describes the amount of added noise needed to make the observables jointly 
measurable. The global joint measurability feature of a probabilistic theory can then be characterized as the 
intersection of all the joint measurability regions associated with the theory. 

We demonstrate that quantum theory contains observables that are as incompatible 
as observables in any probabilistic theory can be. {Hence, we can say that, in a global sense,  quantum theory
has as great a degree of 
incompatibility as any other probabilistic theory. But if only binary observables are considered, the 
degree of incompatibility inherent in quantum theory is limited and we give an example of a probabilistic theory that 
contains maximally incompatible binary observables.}

 Our aim is thus to compare the incompatibility of pairs of observables in different probabilistic physical theories. 
 We first need to set some minimal constraints. 
 
 A probabilistic theory  is a framework that provides
a description of physical systems in terms of states and observables with the following general properties:\\
\\
\indent (i) The states of a system are represented by the elements of a convex subset 
of a real vector space.\\
\indent (ii) An observable is represented as an affine mapping from the set of states into the set of probability distributions 
on some outcome space.
For simplicity, we restrict ourselves here to observables with a finite or countable number of outcomes.\\
\indent (iii)
Any affine mapping from the set of states into the set of probability distributions is a valid observable.\\

We consider a particular probabilistic theory (PT) as given by a family of convex sets of states with associated
sets of observables that share some properties specific to that PT. One may think of each pair consisting of a set
of states with associated set of observables as an {\em instance} of a PT representing a particular type of physical 
system.

Given a PT, we denote by $\M(j | \varrho)$ the probability of obtaining a measurement outcome $j$ when an observable $\M$ is 
measured in a state $\varrho$.
Hence, $0\leq \M(j | \varrho) \leq 1$ and $\sum_j \M(j | \varrho) = 1$.
We will typically label the measurement outcomes by integers.

In quantum theory the states are described by density operators and observables correspond to POVMs \cite{OQP97}.
Their duality is given by the trace formula (with $\varrho$ a density operator and $\sfm$ a POVM)
\begin{equation}\label{eq:povm}
\M(j | \varrho) = \tr{\varrho \Mo(j)} \, .
\end{equation}

Another example of a probabilistic theory is a classical theory, where the states are probability measures on a 
phase space $\Omega$ and observables are traditionally represented as functions $m:\Omega\to\real$; the associated
affine maps from states $\varrho$ to probability distributions are then given by the formula
\begin{equation}
\M(j | \varrho) = \varrho(\{x\in m^{-1}(j)\}) \, .
\end{equation}

Continuing our discussion on  general probabilistic theories, we note that it follows from the required properties 
(i)-(iii) that the set of observables is a convex set; a mixture of two observables is an observable.
Physically mixing corresponds to an experiment where we switch between two measurement apparatuses with 
a random probability.
We can directly write a mixture of two observables with the same set of measurement outcomes.
If the sets of measurement outcomes differ, we can still write a mixture by first adding enough outcomes and then 
embedding both sets into $\integer$.

Another consequence of the basic requirements is that every constant mapping $\varrho\mapsto p$, where $p$ is a 
fixed probability distribution, 
is an observable and we call it a \emph{trivial observable}. 
A trivial observable $\Tmf$ corresponds to a dice rolling experiment, 
where we randomly pick the measurement outcome according to a given fixed probability distribution,
without manipulating the state at all. In quantum theory, trivial 
observables are described by POVMs $\To$ such that each operator $\To(j)$ is a multiple of the identity operator, i.e., 
$\To(j)=t_j \id$ for some $0\leq t_j \leq 1$ with $\sum_j t_j=1$.

The concept of \emph{joint measurement} can be defined in any probabilistic theory.
Two observables $\M_1$ and $\M_2$ are \emph{jointly measurable} if there exists an observable $\M$ such that
\begin{equation}
\sum_k \M(j,k | \varrho) = \M_1(j | \varrho) \, ,
\sum_j \M(j,k | \varrho) = \M_2(k | \varrho) \, .
\end{equation}
In this case $\M$ is called a \emph{joint observable} of $\M_1$ and $\M_2$.
If $\M_1$ and $\M_2$ are not jointly measurable, then we say that they are \emph{incompatible}.

Any probabilistic theory contains jointly measurable pairs of observables.
Namely, a trivial observable $\varrho\mapsto p$ is jointly measurable with any other observable;
we can write a joint observable 
\begin{equation}
\M(j,k | \varrho) = \M_1(j | \varrho) p(k) 
\end{equation}
for the trivial observable and any other observable $\M_1$.
This simply corresponds to an experiment where we measure $\M_1$ and simultaneously roll a dice.
It is a well known fact that, in quantum theory, an observable which is jointly measurable with all other observables is 
necessarily a trivial observable. Indeed, any POVM element of such an observable commutes with all projections  and must therefore be a scalar multiple of the identity (e.g. \cite[Theorem IV.1.3.1]{Ludwig1983}). 

The following simple observation is a key ingredient for our discussion.
\begin{proposition}\label{prop:sum=1}
Let $\M_1$ and $\M_2$ be two observables and $0\leq\lambda\leq 1$.
Then $\lambda \M_1 + (1-\lambda) \Tmf_1$ and $(1-\lambda) \M_2 + \lambda \Tmf_2$ are jointly measurable 
for any choice of trivial observables $\Tmf_1$ and $\Tmf_2$.
\end{proposition}

This proposition can be proved with the following construction.
First, let $p_1$ and $p_2$ be the probability distributions related to $\Tmf_1$ and $\Tmf_2$.
We define an observable $\M$ by formula
\begin{equation}\label{eq:sum=1}
\M(j,k | \varrho) = \lambda p_2(k)\ \M_1(j | \varrho) + (1-\lambda) p_1(j) \ \M_2(k | \varrho) \, .
\end{equation}
For a fixed $\varrho$, the right hand side is clearly a probability distribution.
Moreover, the right hand side is an affine mapping on $\varrho$; therefore $\M$ is an observable.
The marginal observables are
\begin{align*}
\sum_k \M(j,k | \varrho) &= \lambda \M_1(j | \varrho) + (1-\lambda) p_1(j),\\
\sum_j \M(j,k | \varrho) & = (1-\lambda) \M_2(k | \varrho) + \lambda p_2(k).
\end{align*}
This proves Prop. \ref{prop:sum=1}.

The physical idea behind this construction is the following. 
In each measurement run we flip a coin and, depending on the result, we measure either $\M_1$ or $\M_2$ 
in the input state $\varrho$. In this way we get a measurement outcome for either $\M_1$ or $\M_2$.
In addition to this, we roll a dice and pretend that this is a measurement outcome for the other observable.
In this way we get an outcome for both observables simultaneously.
The overall observable is the one given in formula \eqref{eq:sum=1}.

For two observables $\M_1$ and $\M_2$, we denote by $J(\M_1,\M_2)$ the set of all points 
$(\lambda,\mu)\in[0,1]\times[0,1]$ for which there exist trivial observables $\Tmf_1,\Tmf_2$ such that 
$\lambda \M_1 + (1-\lambda) \Tmf_1$ and $\mu \M_2 + (1-\mu) \Tmf_2$ are jointly measurable, and we call
$J(\M_1,\M_2)$ the \emph{joint measurability region} of $\M_1$ and $\M_2$.
The joint measurability region thus characterizes how much noise (in terms of trivial observables) we need to add 
to  {obtain jointly measurable approximations} of $\M_1$ and $\M_2$.
Clearly, $\M_1$ and $\M_2$ are jointly measurable if and only if $(1,1)\in J(\M_1,\M_2)$.

The joint measurability region $J(\M_1,\M_2)$ is a convex region which can be plotted in the plane. 
{To see this,
let $(\lambda',\mu')\in J(\M_1,\M_2)$ and $(\lambda'',\mu'')\in J(\M_1,\M_2)$, then we have to show that
$(\lambda,\mu)\in J(\M_1,\M_2)$ for $(\lambda,\mu)=t(\lambda',\mu')+(1-t)(\lambda'',\mu'')$. Thus let 
$\M_1'=\lambda'\M_1+(1-\lambda')\Tmf_1'$ and $\M_2'=\mu'\M_2'+(1-\mu')\Tmf_2'$ be jointly measurable, and similarly
for $\M_1''=\lambda''\M_1+(1-\lambda'')\Tmf_1''$ and $\M_2''=\mu''\M_2'+(1-\mu'')\Tmf_2''$, with suitable choices
of trivial observables. Then  the observables
$t\M_1'+(1-t)\M_1''$ and $t\M_2'+(1-t)\M_2''$ are jointly measurable \cite[Prop. 2]{BuHe08}.

Note that according to Prop. \ref{prop:sum=1} the line 
$\big\{(\lambda,(1-\lambda))\,:0\,\le\lambda\le 1\big\}\subseteq J(\M_1,\M_2)$. Moreover, it is trivially the case 
that $(0,0)\in J(\M_1,\M_2)$. The convexity of $J(\M_1,\M_2)$ then entails that the convex hull of the three points
$(1,0)$, $(0,1)$ and $(0,0)$ is in $J(\M_1,\M_2)$, hence we have:}
\begin{equation*}
\triangle \equiv \{ (\lambda,\mu)\in [0,1]\times[0,1] : \lambda+\mu \leq 1 \} \subseteq J(\M_1,\M_2) \, .
\end{equation*}

\begin{figure}
\begin{center}
\includegraphics[width=2.5cm]{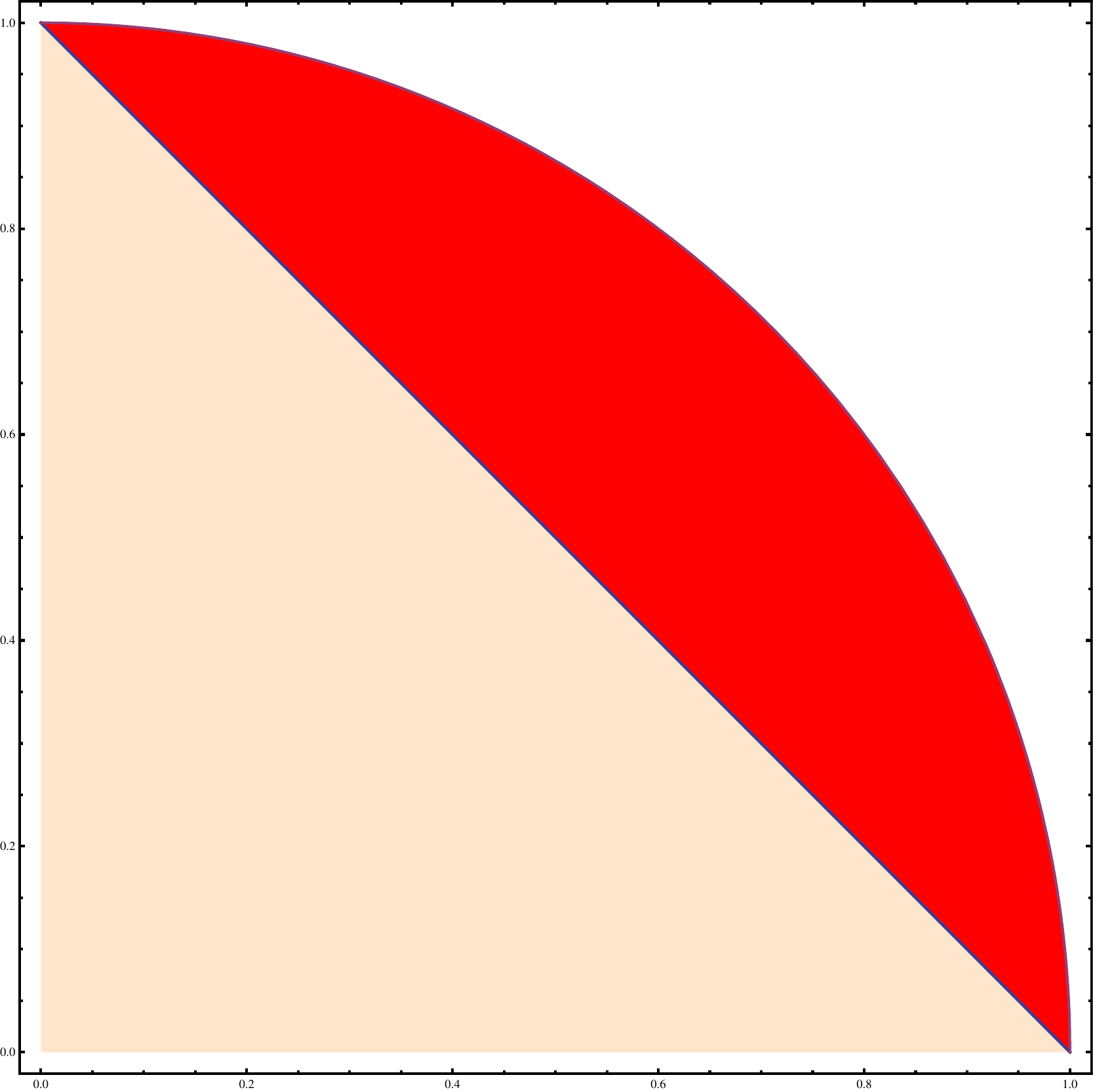}
\end{center}
\caption{(Color online) The region $J(\Mo_x,\Mo_y)$ for two orthogonal spin-$\half$ measurements is a quadrant of the unit disk. 
The region $\triangle$ (light) is a subset of $J(\M_1,\M_2)$ for any pair $\M_1,\M_2$, while the surplus region (dark) depends on the specific pair under consideration.}
\label{fig:qubit}
\end{figure}

As an example, suppose that we are within quantum theory and $\M_1$ and $\M_2$ correspond to spin-$\half$ 
measurements in two orthogonal directions, say $x$ and $y$ -axes.
We then describe them with two POVMs $\Mo_x$ and $\Mo_y$, where
\begin{equation}
\Mo_x(\pm 1) = \half (\id \pm \sigma_x) \, ,\quad \Mo_y(\pm 1) = \half (\id \pm \sigma_y) \, ,
\end{equation}
and $\sigma_x,\sigma_y$ are the usual Pauli matrices in $\complex^2$.
It has been shown in \cite{Busch86} that for the uniformly distributed trivial observable $\pm 1\mapsto\half \id$ (hence describing an unbiased coin), 
the two observables $\lambda \Mo_x + (1-\lambda)\half 1$ and $\mu \Mo_y + (1-\mu) \half 1$ are jointly measurable 
if and only if $\lambda^2+\mu^2\leq 1$. 
{It is also known \cite[Prop. 3]{BuHe08} that} this inequality is a necessary condition for the 
joint measurability of any pair $\lambda \Mo_x + (1-\lambda) \To_1$ and $\mu \Mo_y + (1-\mu) \To_2$, where 
$\To_1,\To_2$ are arbitrary trivial observables. 
Therefore, we conclude that
\begin{equation}
J(\Mo_x,\Mo_y)=\{ (\lambda,\mu)\in[0,1]\times[0,1] : \lambda^2+\mu^2\leq 1\} \, .
\end{equation}
This region is depicted in Fig.~\ref{fig:qubit}.

In addition to describing the incompatibility of pairs of observables, the concept of a joint measurability region also provides a means to compare 
{the degrees of incompatibility inherent in entire theories.} 
A global joint measurability feature of a probabilistic theory PT is characterized by the intersection of all the sets 
$J(\M_1,\M_2)$ across all instances of PT, and we denote
\begin{align*}
J_{PT} =  & \{ (\lambda,\mu)\in [0,1]\times[0,1] : (\lambda,\mu)\in J(\M_1,\M_2) \\
&   \textrm{for all pairs of observables $\M_1$ and $\M_2$}\\
& \textrm{in all instances of PT} \}.
\end{align*}
We call $J_{PT}$ the \emph{joint measurability region} 
{for} $PT$.
We always have $\triangle\subseteq J_{PT}$, but $J_{PT}$ can be larger than $\triangle$.
The larger the surplus region is, the more jointly measurable the theory is globally; see Fig.~\ref{fig:regiong}.
If $(\lambda,\mu)\notin J_{PT}$, this means that there is a pair of observables $\M_1$ and $\M_2$ such that the 
mixtures $\lambda \M_1 + (1-\lambda) \Tmf_1$ and $\mu \M_2 + (1-\mu) \Tmf_2$ are incompatible with any choice 
of trivial observables $\Tmf_1$ and $\Tmf_2$.

\begin{figure}
    \centering
    \subfigure[]
    {
        \includegraphics[width=2.5cm]{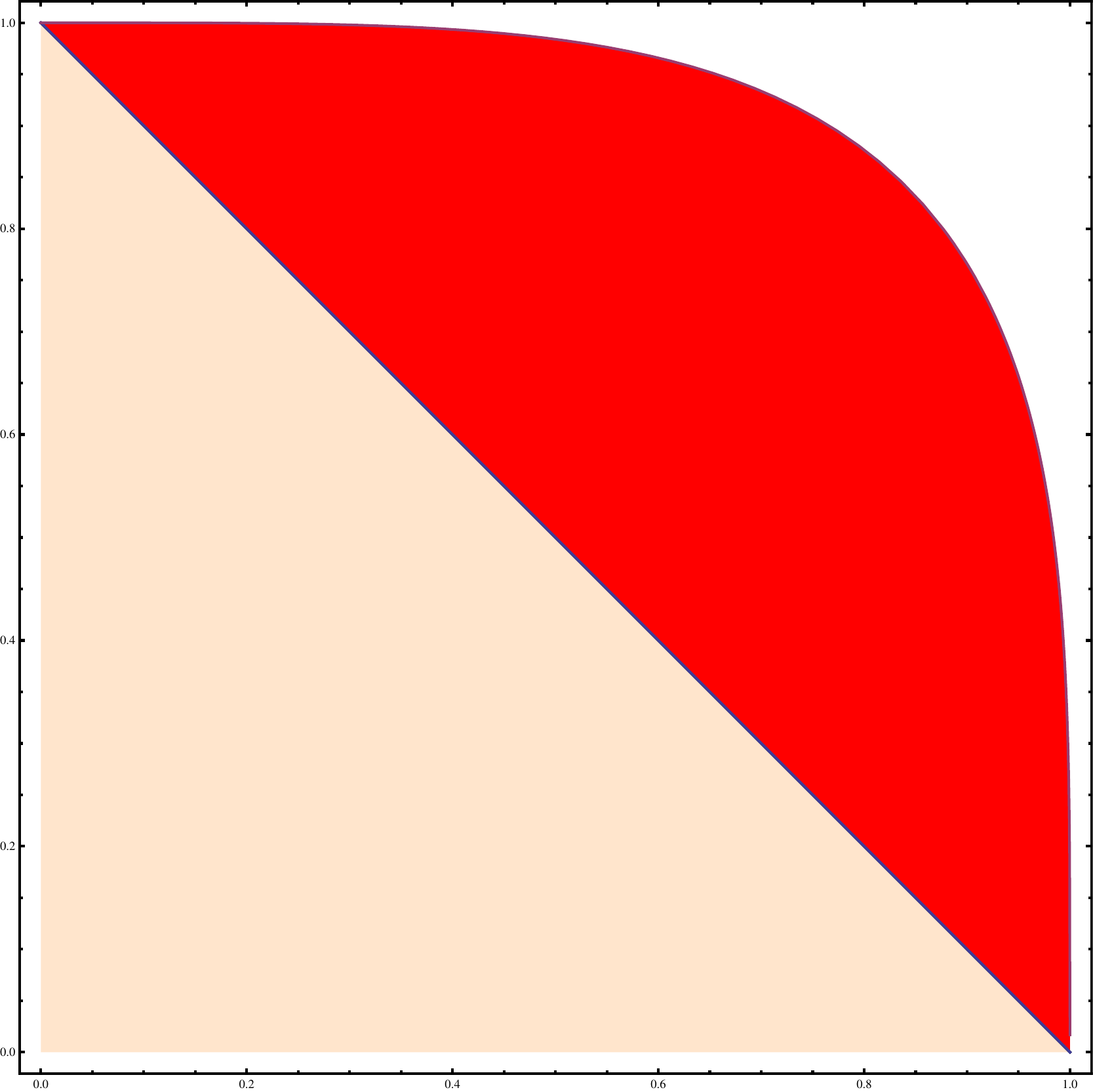} 
    }
    \subfigure[]
    {
        \includegraphics[width=2.5cm]{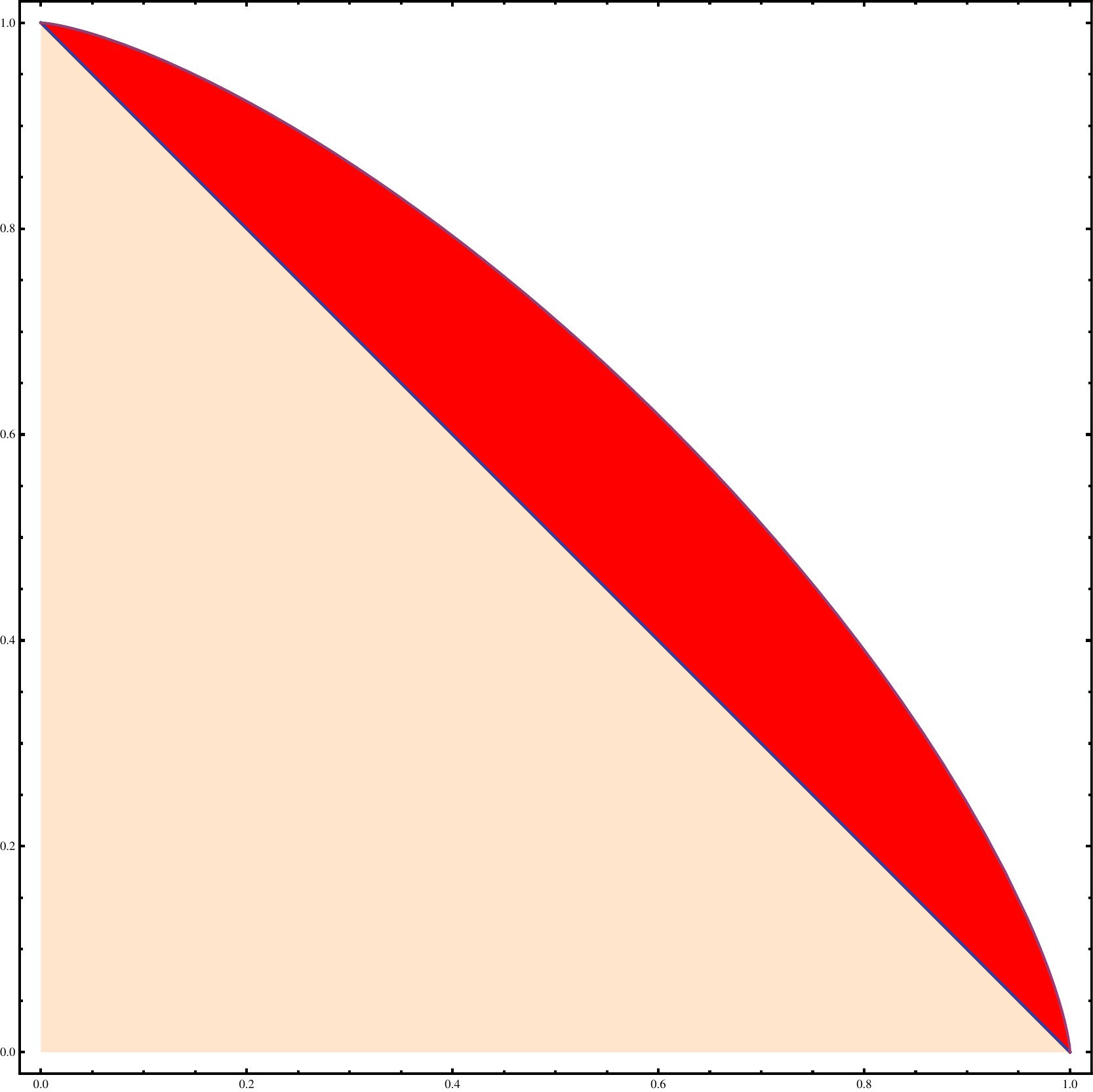}   
    }
        \caption{(Color online) The region $\triangle$ (light) is a subset of the joint measurability region $J_{PT}$ (colored) 
{for} any probabilistic theory.
The larger the surplus region (dark) is, the more jointly measurable the theory globally is.
If (a) and (b) 
{are} joint measurability regions for two different probabilistic theories, then we 
{can conclude that (b) represents a greater degree of incompatibility than (a).}}
    \label{fig:regiong}
\end{figure}

Since $J_{PT}$ can be defined in any probabilistic theory, we can compare the joint measurability regions 
{for} different theories.
We obviously have $J_{PT}=[0,1]\times [0,1]$ in any probabilistic theory where all measurements are jointly measurable, such as the classical probability theory.
In the case of the greatest degree of incompatibility we would have $J_{PT}=\triangle$.
We will next show that quantum theory 
{incorporates, globally, as much incompatibility between pairs of observables 
as a probabilistic theory can do.}

\begin{theorem}\label{th:main}
In quantum theory $J_{QT}= \triangle$.
In particular, $J_{QT} \subseteq J_{PT}$ for any probabilistic theory PT.
\end{theorem}

In quantum theory every observable $\M$ corresponds to a unique POVM $\Mo$ by equation \eqref{eq:povm}.
We will prove that for any pair $(\lambda,\mu)\notin\triangle$, there are quantum observables $\Mo_1$ and $\Mo_2$ 
such that the mixtures $\lambda \Mo_1 + (1-\lambda) \To_1$ and $\mu \Mo_2 + (1-\mu) \To_2$ are incompatible 
with any choice of trivial observables $\To_1,\To_2$.
Our proof is based on a recent result \cite{CaHeTo12} on the joint measurability region for two complementary observables, 
which is a generalization of the result illustrated in Fig.~\ref{fig:qubit}.

\begin{proof}[Proof of Theorem \ref{th:main}]
We have earlier seen that $\triangle\subseteq J_{PT}$, so we need to show that $J_{QT}\subseteq\triangle$.
Let $(\lambda,\mu)\notin\triangle$, i.e., $\lambda+\mu >1$.
Fix $\epsilon>0$ such that $\lambda+\mu >1+\epsilon$.
We then choose $d$ to be a positive integer satisfying 
\begin{equation}
 \frac{\sqrt{d}-1}{d-1} \leq \epsilon \, .
\end{equation}
(This can be done since the left hand side $\to 0$ when $d\to\infty$.)
We will consider a quantum system that is described by a $d$-dimensional Hilbert space $\hi_d$.
Let $\{\varphi_j\}_{j=0}^{d-1}$ be an orthonormal basis for $\hi_d$.
We define another orthonormal basis  $\{\psi_k\}_{k=0}^{d-1}$ for $\hi_d$ by 
\begin{equation}\label{eq:fourier}
\psi_k=1/\sqrt{d} \sum_j e^{-2\pi i \frac{j k}{d}} \varphi_j \, .
\end{equation}
The orthonormal bases $\{\varphi_j\}_{j=0}^{d-1}$ and $\{\psi_k\}_{k=0}^{d-1}$  are mutually unbiased, i.e., 
$\mo{\ip{\varphi_j}{\psi_k}}=constant$.
We define two POVMs $\Mo_1$ and $\Mo_2$ by
\begin{equation}
\Mo_1(j) = \kb{\varphi_j}{\varphi_j} \, , \quad \Mo_2(k) = \kb{\psi_k}{\psi_k} \, .
\end{equation}
We thus obtain a pair of $d$-outcome observables on $\hi_d$.
Since $\Mo_1$ and $\Mo_2$ consist of projections and $\Mo_1(j)\Mo_2(k)\neq \Mo_2(k)\Mo_1(j)$, 
it follows that they are incompatible.

As proved in \cite{CaHeTo12}, the observables $\lambda' \Mo_1 + (1-\lambda') \To_1$ and 
$\mu' \Mo_2 + (1-\mu') \To_2$ are incompatible for any choice of trivial observables $\To_1,\To_2$ whenever
\begin{equation}
\lambda' + \mu' > 1+ \frac{\sqrt{d}-1}{d-1}  \, .
\end{equation}
Since 
\begin{equation}
\lambda + \mu > 1+ \epsilon \geq 1+ \frac{\sqrt{d}-1}{d-1} \, ,
\end{equation}
 we conclude that $(\lambda,\mu)\notin J_{QT}$.
\end{proof}

Using the ideas of the proof of Theorem 1, we can also show that the conclusion $J_{QT}=\triangle$ can be reached by using a \emph{single pair} of incompatible 
observables if we consider an infinite dimensional system and observables with a countably infinite number of outcomes.

Let $\hi$ be an infinite dimensional Hilbert space and write it as a direct sum of finite $d$-dimensional Hilbert spaces $\hi_d$, 
$\hi=\bigoplus_{d=2}^\infty \hi_d$. In each $\hi_d$ consider a pair of mutually unbiased 
orthonormal bases $\{\varphi^d_{j}\}_{j=0}^{d-1}$ and  $\{\psi^d_{k}\}_{k=0}^{d-1}$, where the latter is obtained  
from the first one by the formula \eqref{eq:fourier}.
We define two POVMs $\No_{1}$ and $\No_{2}$ via
\begin{equation}\label{eqn:totalmub}
\No_{1}(d,j)=\kb{\varphi^d_{j}}{\varphi^d_{j}} \, , \quad \No_{2}(d,k)=\kb{\psi^d_{k}}{\psi^d_{k}} \, .
\end{equation}
These observables act in the infinite dimensional Hilbert space $\hi$ and $d$ in \eqref{eqn:totalmub} is an index labeling 
the different outcomes.
The outcome space of $\No_1$ and $\No_2$ is $\Omega_\infty\equiv \{(d,j):d=2,3,\ldots, j=0,\ldots,d-1\}$.

\begin{theorem}\label{th:totalmub}
The observables $\No_1$ and $\No_2$ defined in \eqref{eqn:totalmub} satisfy $J(\No_1,\No_2)=\triangle$.
\end{theorem}

\begin{proof}[Proof of Theorem \ref{th:totalmub}]
Let $p_1$ and $p_2$ be two probability distributions defined on $\Omega_\infty$. 
Assume that $\lambda+\mu>1$ and define two observables $\No_{1,\lambda},\No_{2,\mu}$ via
\begin{equation}
\begin{array}{lcl}
\No_{1,\lambda}(d,j)&=&\lambda\, \kb{\varphi^d_{j}}{\varphi^d_{j}}+(1-\lambda)\,p_1(d,j)I,\\
 \No_{2,\mu}(d,k)&=&\mu\,\kb{\psi^d_{k}}{\psi^d_{k}}+(1-\mu)\,p_2(d,k)I \, .
 \end{array}
\end{equation}

We need to show that $\No_{1,\lambda}$ and $\No_{2,\mu}$ are incompatible.
To prove this, we make
{the} counter assumption that $\No_{1,\lambda},\No_{2,\mu}$ are jointly measurable. 
This implies that for any projection $P$ on $\hi$, the projected observables $P\No_{1,\lambda}P$ and 
$P\No_{2,\mu}P$ acting on a subspace $P\hi$ are jointly measurable.
 (If $\sfg$ is a joint observable of  two observables $\Mo_1,\Mo_2$, then $P\sfg P$ is a joint observable of $P\Mo_1P,P\Mo_2P$ in $P\hi$.)
Especially, the projections of $\No_{1,\lambda}$ and $\No_{2,\mu}$ to any subspace $\hi_d$ should be jointly measurable. 
But from the result cited in the proof of Theorem \ref{th:main} we know that for $d$ large enough, the projections to $\hi_d$ are incompatible.
Hence, $\No_{1,\lambda}$ and $\No_{2,\mu}$ are incompatible.
\end{proof}

We note that the observables $\No_1$ and $\No_2$ defined in \eqref{eqn:totalmub} are not the only pair satisfying $J(\No_1,\No_2)=\triangle$.
Namely, we can modify $\No_1$ and $\No_2$ in any chosen subspace $\hi_d$ but the conclusion $J(\No_1,\No_2)=\triangle$ is still true since it depends on the fact that $\No_1$ and $\No_2$ contain mutually unbiased bases in arbitrarily high dimension.

An interesting problem within quantum theory would be to try to find a characterization of all pairs of quantum observables $\Mo_1,\Mo_2$ that satisfy $J(\Mo_1,\Mo_2)=\triangle$. In particular, we may ask if maximally incompatible observables can exist in a finite dimensional Hilbert space, or if they can have a finite number of outcomes. 
 Since two mutually unbiased bases are expected to be among the most incompatible observable pairs
in a fixed dimension $d$, our construction in the proof of Theorem \ref{th:main}  suggests that the answer to the first question would be negative.
 A proof of this 
 {claim} is, however, lacking.

As for the second question, we can present a  partial answer by investigating the joint measurability region in the case of pairs of 
binary quantum observables. Our aim is to show that
$$
\{ (\lambda,\mu)\in[0,1]\times[0,1] : \lambda^2+\mu^2\leq 1\}  \subseteq J (\Mo_1, \Mo_2)
$$
for any binary observables $\Mo_1$ and $\Mo_2$, regardless of the dimension of the Hilbert space. In other words, we will show that two orthogonal spin observables are as incompatible as any binary observables can be. 

To this end, let us note that two binary quantum observables are incompatible if and only if they enable a violation of the Bell-CHSH inequality \cite{WoPeFe09}. We must therefore look at the Bell expression
$$
\mathcal{B} = \vert \langle \Mo_{1} \No_{1}\rangle  + \langle \Mo_{1} \No_{2}\rangle + \langle \Mo_{2} \No_{1}\rangle - \langle \Mo_{2}\No_{2}\rangle \vert\, .
$$
Let us denote $\alpha = \langle \Mo_{1} \No_{1}\rangle  + \langle \Mo_{1} \No_{2}\rangle $ and $\beta = \langle \Mo_{2} \No_{1}\rangle - \langle \Mo_{2}\No_{2}\rangle$. {By \cite[Theorem 1]{Tsirelson1980}, there exist unit vectors $\mathbf{x}_1,\mathbf{x}_2, \mathbf{y_1},\mathbf{y}_2\in\mathbb{R}^4$ such that $\langle \Mo_{j} \No_{k}\rangle   = \mathbf{x}_j \cdot \mathbf{y}_k$ for $j,k=1,2$;  and conversely, given any quadruple of unit vectors there exist a corresponding set of binary observables and a bipartite state such that this equality holds. In particular,  we have $\alpha = \mathbf{x}_1 \cdot (\mathbf{y}_1 + \mathbf{y}_2)$  and $ \beta = \mathbf{x}_2\cdot (\mathbf{y}_1 - \mathbf{y}_2)
$ so that an application of the Cauchy-Schwarz inequality along with the parallelogram law yields
\begin{eqnarray*}
\alpha^2 + \beta^2 &\leq& \Vert \mathbf{x}_1\Vert^2\Vert \mathbf{y}_1 + \mathbf{y}_2\Vert^2 + \Vert \mathbf{x}_2\Vert^2\Vert \mathbf{y}_1 - \mathbf{y}_2\Vert^2\\
& =  & 2\Vert \mathbf{y}_1 \Vert^2 + 2\Vert \mathbf{y}_2\Vert^2 = 4.
\end{eqnarray*}
By choosing the unit vectors appropriately we also see that any pair $(\alpha,\beta)$ satisfying this condition can be obtained.}

 If we now mix the observables $\Mo_j$ with the trivial observable $\To (\pm 1)= \frac{1}{2} I  $ with some weights $\lambda$ and $\mu$ we see that the pair $(\alpha, \beta)$ turns into $(\lambda \alpha, \mu \beta)$, thus changing the Bell expression from $\vert \alpha +\beta \vert$ to $\vert \lambda \alpha  + \mu \beta \vert$. We must therefore determine those $(\lambda, \mu)$ for which $\vert \lambda \alpha + \mu \beta \vert \leq 2 $ for all $(\alpha, \beta)$ satisfying $\alpha^2 + \beta^2 \leq 4$ (see Fig. \ref{fig:bell}). But the boundary curve for this region is obtained when the equations $ \left(\alpha/\lambda\right)^2 + \left( \beta /\mu \right)^2 = 4$ and $\alpha + \beta  = 2$ have at most one common solution. By inserting $\beta = 2-\alpha$ into the first equation the problem reduces to determining when the discriminant is negative or zero, and one readily verifies that this is the case exactly when $\lambda^2 + \mu^2  \leq 1$. 

\begin{figure}
    \centering
    \subfigure[]
    {
        \includegraphics[width=4cm]{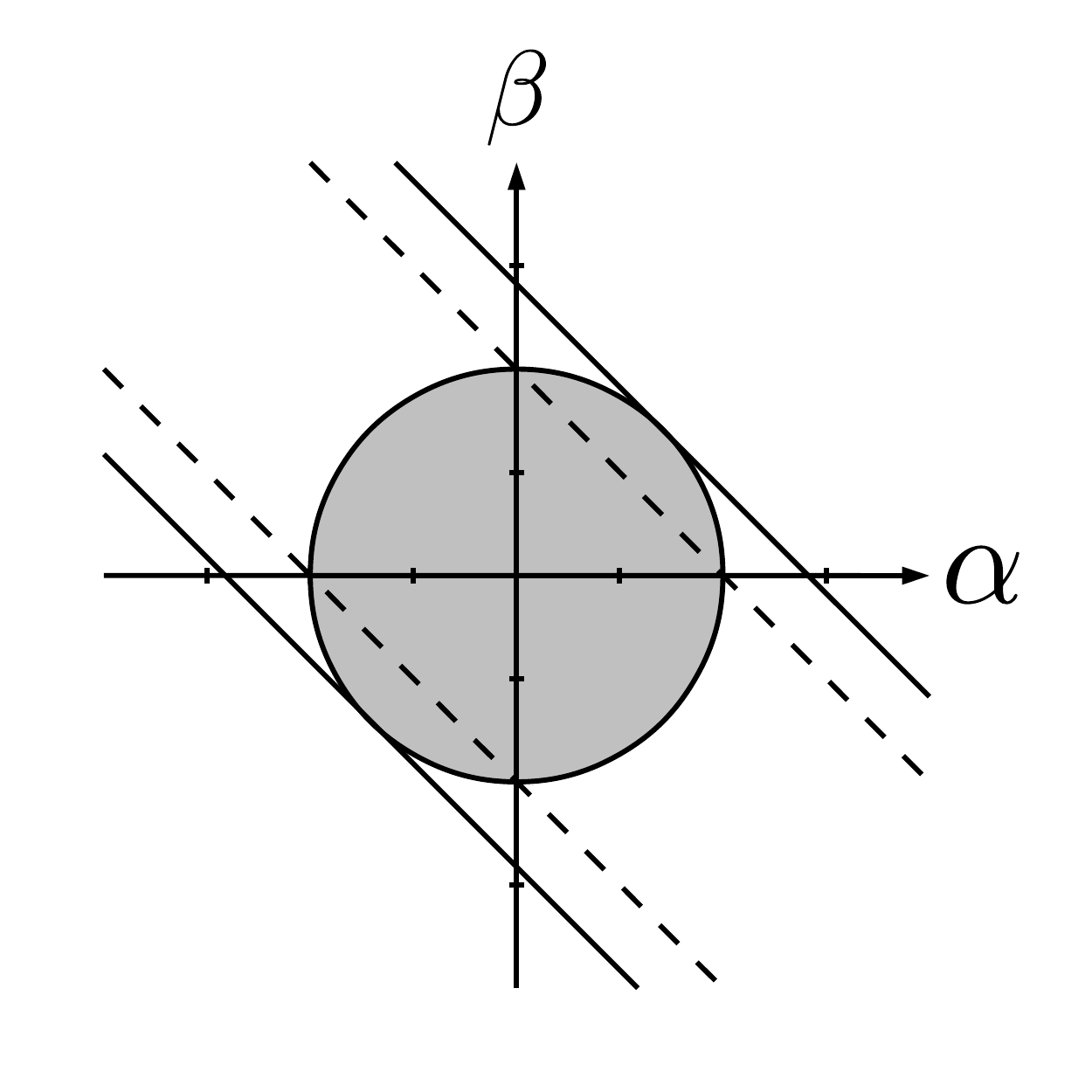}
        \label{fig:bell1}
    }
    \subfigure[]
    {
        \includegraphics[width=4cm]{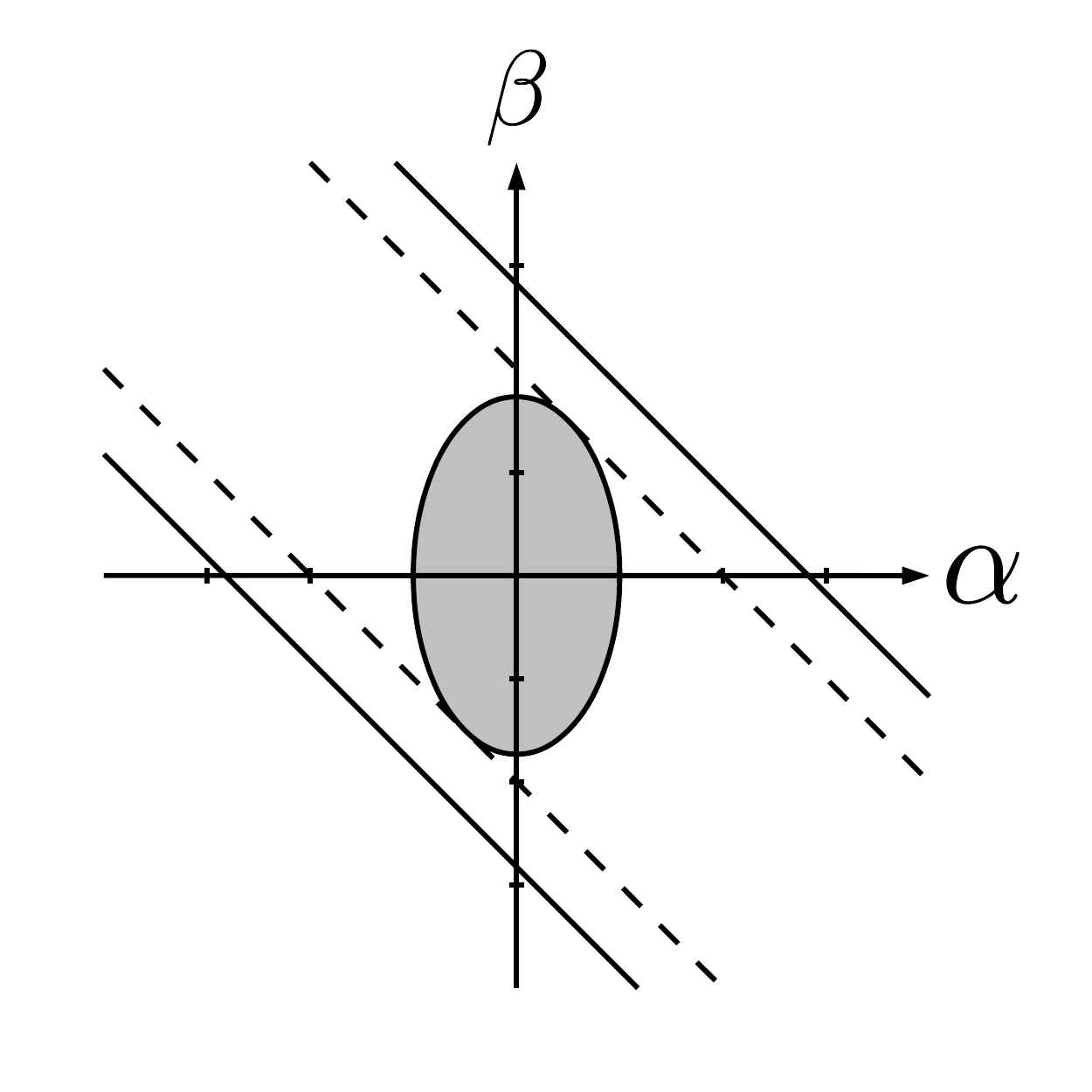}
        \label{fig:bell2}
    }
        \caption{In (a) the grey area represents the possible values that $\alpha$ and $\beta$ can obtain by varying the observables and the state in the Bell expression $\mathcal{B} =\vert \alpha + \beta\vert$. The solid lines represent the Tsirelson bound $\mathcal{B}  = 2\sqrt{2}$ and the dashed lines represent the bound $\mathcal{B}  =2$. By considering only observables which are mixtures with the uniformly distributed trivial observable with fixed $\lambda$ and $\mu$, the area becomes smaller as depicted in (b), and  a suitable choice of weights makes the violation of the Bell-CHSH inequality impossible.}
    \label{fig:bell}
\end{figure}

In conclusion, given any pair of binary observables $\Mo_1$ and $\Mo_2$, and weights $\lambda$ and $\mu$ with $\lambda^2 + \mu^2\leq 1$, the mixtures $\lambda\Mo_1 + (1-\lambda)\To$ and $\mu\Mo_2  + (1-\mu) \To$ can not be used to violate the Bell-CSHS inequality and must therefore be jointly measurable. We note that in the case $\mu = \lambda$ the same result using a different technique has been obtained by Banik {\em et al.} \cite{BaGaGhKa12}.

Although Theorem \ref{th:totalmub} shows that quantum theory contains pairs of observables that are maximally incompatible, the strictly larger joint measurability region when restricting to binary observables suggests that 
{more fine grained quantifications of 
the global degree of  incompatibility between observables might not rank quantum theory 
among the most extreme theories
in this respect.} The example below will show that when restricting to just binary observables, it is indeed possible for a theory to have the smallest possible joint measurability region. 
{In that sense such a theory must be considered to embody a strictly
greater degree of incompatibility than quantum theory.}

Consider any probabilistic theory, which contains a state space isomorphic to a square, by which we mean the convex hull of {four different points $s_1,s_2,s_3,s_4$ in $\real^2$ satisfying $s_1+s_4=s_2+s_3$, for instance $s_1=(0,0), s_2=(0,1), s_3=(1,0)$ and $s_4=(1,1)$.}
 We will show that there is a pair of binary observables which are maximally incompatible. Let $\M_1$ and $\M_2$ be binary observables that pick out the right and top sides of the square respectively, i.e.
\begin{equation}
\begin{array}{c}
\M_1(+|s_1) = \M_1(+|s_2) = 0, \\
\M_1(+|s_3) = \M_1(+|s_4) = 1, \\
\M_2(+|s_1) = \M_2(+|s_3) = 0, \\
\M_2(+|s_2) = \M_2(+|s_4) = 1. 
\end{array}\label{eq:max}
\end{equation}

\begin{proposition}\label{lem:squit}
For the binary observables $\M_1$ and $\M_2$ defined in \eqref{eq:max}, $J(\M_1,\M_2)=\triangle$.
\end{proposition}

\begin{proof}
Suppose that there exists a joint observable $\M$ for $\lambda\M_1 + (1-\lambda) \Tmf_1$ and $\mu\M_2 + (1-\mu)\Tmf_2$ where $\Tmf_1$ and $\Tmf_2$ are trivial observables. Let $p_1$ and $p_2$ be the probability distributions associated to  $\Tmf_1$ and $\Tmf_2$ so that we have for any state $\rho$
$$
\begin{array}{c}
\M(+,+|\varrho) + \M(+,-| \varrho) = \lambda\M_1(+|\varrho) + (1-\lambda)p_1(+) \\
\M(-,+|\varrho) + \M(-,-|\varrho) = \lambda\M_1(-|\varrho) + (1-\lambda)p_1(-) \\
\M(+,+|\varrho) + \M(-,+|\varrho) = \mu\M_2(+|\varrho) + (1-\mu)p_2(+) \\
\M(-,-| \varrho) + \M(+,-|\varrho) = \mu\M_2(-|\varrho) + (1-\mu)p_2(-)
\end{array}
$$

Any $\M$ satisfying such marginal properies will be correctly normalised, but to be a valid observable, all the components of $\M$ must take positive values on the points $s_i$.
In particular, we must have
\begin{align*}
\M(+,-|s_2) =& (1-\lambda) p_1(+) - \M(+,+|s_2) \ge 0, \\
\M(-,+|s_3) =& (1-\mu) p_2(+) - \M(+,+|s_3) \ge 0, \\
\M(-,-|s_4) =&  1 +\M(+,+|s_4) -\lambda-(1-\lambda)p_1(+)\\
&-\mu-(1-\mu)p_2(+)  \ge 0.
\end{align*}
Rewriting the last of these inequalities and invoking the defining property on the $s_i$ gives
\begin{align*}
\lambda+\mu \le& 1-(1-\lambda)p_1(+)-(1-\mu)p_2(+)+\M(+,+|s_4) \\
\le&-(1-\lambda)p_1(+)+\M(+,+|s_2) \\
&-(1-\mu)p_2(+)+\M(+,+|s_3) \\
& +1-\M(+,+|s_1)\le 1,
\end{align*}
where the final step comes about from invoking the positivity of $\M$ on $s_1$.
\end{proof}

The result of Proposition 2 does not come as a surprise in light of the fact that the barrier to maximal incompatibility of binary quantum observables comes from the connection with a Bell-CHSH inequality. Indeed, square shaped state spaces have been used in a model of a probabilistic theory containing the PR boxes which violate such an inequality to its maximal possible value.

We note that the conclusion of Proposition 2 is not restricted to the {square state space}. Consider any state space containing a {square} whose 
vertices $s_i$ are extreme points of the state space and whose boundary lines lie on the boundary of the state space; assume further
that opposite sides of the square are contained in parallel hyperplanes that do not intersect with the interior of the state space. These two pairs 
of hyperplanes define effects whose values on the $s_i$ satisfy Eq.~(\ref{eq:max}). It follows that the proof of Proposition 2 can be adopted
in such cases. Examples are given by state space of the following shapes: pyramid, double pyramid, cube, cylinder.

{The fact that the restriction to just binary observables allows one to differentiate between probabilistic theories that both contain 
maximally incompatible observables suggests that a more fine grained global measure of the degree of incompatibility is 
needed if the aim is to pick out a single theory as the one containing overall the most incompatible pairs of observables.}
For instance, for  a given probabilistic theory PT we may define   $J_{PT}^{(d)}$ to be the joint measurability region for all 
possible $d$-outcome observables in PT. Since increasing the number of outcomes of observables by simply adding outcomes 
that never occur does not change the properties of incompatibility, we immediately have 
$J_{PT}^{(d+1)} \subseteq J_{PT}^{(d)}$. By comparing the regions in different theories for different values of  $d$ we obtain 
a more fine grained way of comparing the 
{degrees of incompatibility within the theories. It may even turn out that in this sense
quantum theory embodies globally the least amount of incompatibility among the theories containing maximally incompatible 
observables.} However, this is still an open question and a topic for future investigations.


\

\noindent
{\bf Acknowledgements.} 
The authors wish to thank Tom Bullock for comments on an earlier version of this paper.
T.H. and J.S. acknowledge financial support from the Academy of Finland (grant no. 138135).  
J.S. also acknowledges financial support from the Italian Ministry of Education, University and Research (FIRB project RBFR10COAQ). 
{N.S. gratefully acknowledges support through the award of an Annie Currie Williamson PhD Bursary
at the University of York. Finally, the authors wish to thank an anonymous referee for suggesting a number of improvements to the
presentation of the material.}


\end{document}